\documentclass[10pt, conference]{IEEEtran}
\ifCLASSINFOpdf
\else
\fi
\usepackage{setspace}
\usepackage{verbatim}
\usepackage{cite}
\usepackage{amsmath,amssymb,amsfonts,amsthm}
\usepackage{algorithm}
\usepackage{graphicx}
\usepackage{comment}
\usepackage{textcomp}
\usepackage{color}
\usepackage{stfloats}

\theoremstyle{plain}

\theoremstyle{definition}

\newtheorem{lemma}{Lemma}

\theoremstyle{remark}

\usepackage{hyperref}

\def\BibTeX{{\rm B\kern-.05em{\sc i\kern-.025em b}\kern-.08em
    T\kern-.1667em\lower.7ex\hbox{E}\kern-.125emX}}
\usepackage{bm}      
\usepackage{enumitem}
\usepackage{algpseudocode}
\usepackage{amsmath,amssymb}

\newcommand\blfootnote[1]{%
  \begin{NoHyper}%
  \renewcommand\thefootnote{}\footnote{#1}%
  \addtocounter{footnote}{-1}%
  \end{NoHyper}%
}

\begin{document}

\title{\huge Joint Device Pairing and Bandwidth Allocation Optimisation for Semantic Feature Multiple Access Networks}

\author{
\IEEEauthorblockN{Jiaxiang Wang\IEEEauthorrefmark{1},
Zhaohui Yang\IEEEauthorrefmark{2}\IEEEauthorrefmark{3},
Mingzhe Chen\IEEEauthorrefmark{4},
and Mohammad Shikh-Bahaei\IEEEauthorrefmark{1}}
	\IEEEauthorblockA{
 		$\IEEEauthorrefmark{1}$Department of Engineering, 
 King's College London,
London, UK \\
			$\IEEEauthorrefmark{2}$College of Information Science and Electronic Engineering, Zhejiang University, Hangzhou, China\\
   	$\IEEEauthorrefmark{3}$Zhejiang Provincial Key Laboratory of Info. Proc., Commun. \& Netw. (IPCAN), Hangzhou, China\\
        $\IEEEauthorrefmark{4}$ Department of Electrical and Computer Engineering and Institute for Data Science and Computing, University of Miami\\
            E-mails: 
jiaxiang.wang@kcl.ac.uk,
yang\_zhaohui@zju.edu.cn,
mingzhe.chen@miami.edu,
m.sbahaei@kcl.ac.uk
		}}

\maketitle

\begin{abstract}
This paper presents a Semantic Feature Multiple Access (SFMA) framework for multi-user semantic communication in downlink wireless systems. By extending SwinJSCC to a two-user superimposition paradigm, SFMA enables simultaneous semantic transmission to multiple users over shared time-frequency resources. A key innovation is the Cross-User Attention (CUA) module, which facilitates controlled semantic feature exchange between paired users by leveraging inter-image similarity while mitigating interference. We formulate a joint user pairing and resource allocation problem to minimize global semantic distortion under constraints on bandwidth, end-to-end latency, and energy. This mixed-integer non-convex problem is decomposed into a Minimum-Weight Perfect Matching (MWPM) sub-problem and a convex bandwidth allocation feasibility check, with semi-closed-form bandwidth bounds derived from a strictly concave rate expression. A polynomial-time algorithm based on Blossom matching and bisection search is proposed. Extensive simulations on ImageNet-100 show that SFMA significantly improves reconstruction quality across pairing modes, and the proposed optimization effectively reduces overall distortion while satisfying physical-layer constraints.

\blfootnote{This work was supported in part by the U.S. National Science Foundation under Grant SaTC-2520900.}
\end{abstract}

\vspace{-6mm}\section{Introduction}
The advent of 6G is driving a paradigm shift from traditional bit-centric communication to semantics-driven systems \cite{saad2019vision,wang2022performance,yang2025privacy,yang2023energy,wei2025optimizing,xu2025transformer}, where the goal is meaning recovery rather than raw bit fidelity. In this context, deep learning enabled joint source–channel coding (JSCC) learns an end-to-end mapping between source and channel symbols, achieving robust performance under low SNRs \cite{bourtsoulatze2019deep}. While single-user semantic (or deep JSCC) systems have received intensive attention, semantic communications in multi-user multiple access (MA) settings remain under-explored. Conventional multiple access schemes such as Non-Orthogonal Multiple Access (NOMA) \cite{li2023non} and Rate Splitting Multiple Access (RSMA) \cite{zhao2025compression} operate in the signal or symbol domain and focus on interference management, thus failing to exploit semantic feature correlations among users. In contrast, the proposed Semantic Feature Multiple Access (SFMA) operates in the semantic feature domain, improving bandwidth efficiency under limited resources. However, designing SFMA-based communication schemes must address key challenges: (i) semantic feature interaction—how to exploit redundancy among correlated user features; (ii) resource allocation under semantics—jointly considering semantic distortion, feature correlation, user pairing, and bandwidth/power/time allocation; (iii) scalable algorithmic design—since the joint problem is highly non-convex, efficient methods are required.

Recent research \cite{wang2025semantic,bo2024deep} on multi-user semantic communications has explored a variety of system architectures and optimization strategies. Wang et al. \cite{wang2025semantic} investigated semantic interference in multi-user resource allocation, modeling the semantic interference factor as a function of both transmit power and signal-to-noise ratio (SNR). Meanwhile, Bo et al. \cite{bo2024deep} introduced DeepSCM, a deep learning-based superposition coded modulation scheme that hierarchically encodes semantic information into basic and enhanced feature streams. 
Despite these advances, 
a systematic methodology for the joint optimization of user pairing and resource allocation in multi-user semantic communication systems remains largely unaddressed.
The main contribution of this work is a novel SFMA framework that enables similarity-conditioned cross-user semantic feature fusion, and user-pairing optimisation for multi-user downlink transmission. Our key contributions include:
\begin{itemize}
    \item We propose a novel SFMA framework that enables semantic-level superposition for downlink multi-user semantic communications. By grouping users into pairs and combining their encoded features in the semantic domain, SFMA allows simultaneous transmission over shared time-frequency resources, significantly improving bandwidth efficiency and reconstruction quality.
    \item We design a Cross-User Attention (CUA) module that leverages inter-user semantic correlation to guide feature fusion. Incorporating windowed and shifted-window cross-attention with a cosine-similarity gate, CUA adaptively controls semantic exchange, enhancing feature reuse while mitigating interference, without additional bandwidth cost.
    \item We formulate a joint user pairing and bandwidth allocation problem to minimize total semantic distortion under bandwidth, power, latency, and energy constraints. The mixed-integer non-convex problem is decomposed into a Minimum-Weight Perfect Matching (MWPM) problem and a convex feasibility subproblem. We derive closed-form bandwidth bounds and develop a polynomial-time algorithm using Blossom matching and bisection search, ensuring both optimality and scalability.
\end{itemize}

\vspace{-3mm}

\section{System Model}
\vspace{-2mm}
\addtolength{\topmargin}{0.03in}

\begin{figure}[!t]
    \centerline{\includegraphics[width=0.5\textwidth]{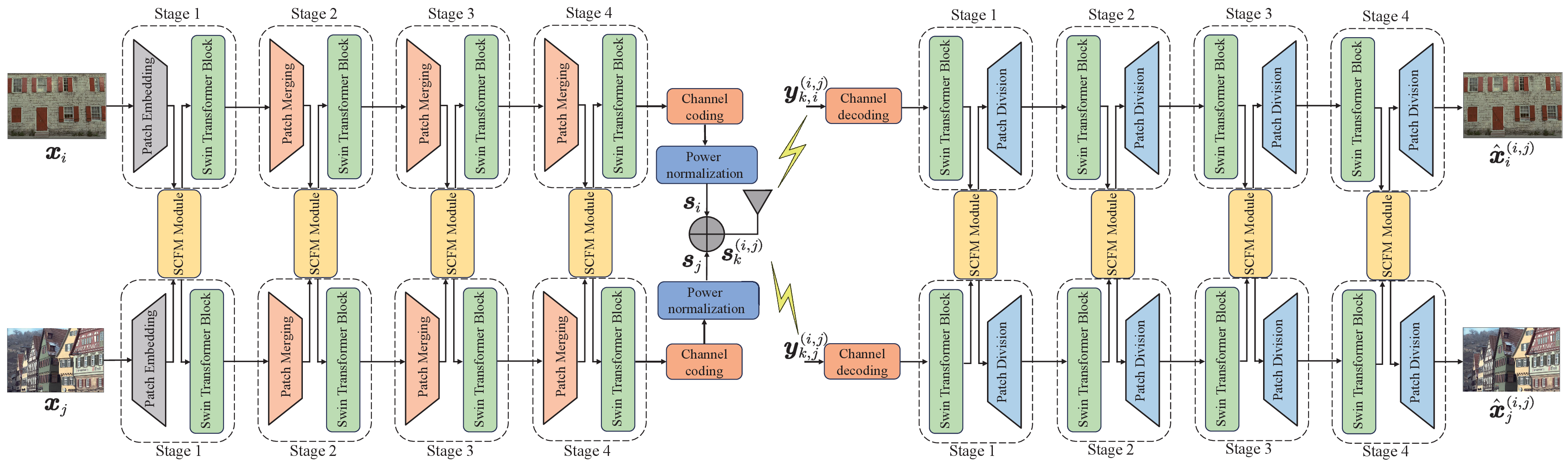}}
    \caption{System architecture of the proposed SFMA framework.}
    \label{fig: System model}
\end{figure}

We consider a downlink semantic communication system where one BS serves a set $\mathcal{N}$ of $N$ users using semantic communication techniques. 
The users are divided into a set $\mathcal{K}$ of $K$ user groups and each group includes two users, where $K=N/2$. The BS implements semantic feature based multiple access (SFMA) \cite{wang2025generative} method to combine the data of the two users $i,j \in \mathcal{N}$ in a group $k\in\mathcal{K}$ in the semantic domain. Then, the BS transmits the semantic information to the two users in a group using the same time and frequency resource. The users use semantic decoders to process the received semantic information and regenerate the data that the BS transmits. 
\vspace{-3mm}
\subsection{Semantic Feature-based Multiple Access Method}
\vspace{-1mm}

\vspace{-1mm}
For each user $u \in \{i,j\}$, the input image $\bm{x}_u \in \mathbb{R}^{H \times W \times 3}$ undergoes patch embedding to generate the initial token sequence $\bm{x}_u^1 \in \mathbb{R}^{L_1 \times d_1}$, where $L_1 = \frac{H}{2} \times \frac{W}{2}$ and $d_1$ is the feature dimension.
The encoder comprises four hierarchical stages that progressively transform the token sequence. 
At Stage 1, the input $\bm{x}_u^1 \in \mathbb{R}^{L_1 \times d_1}$ is processed through the CUA module and $N_1$ Swin Transformer blocks, producing the intermediate output $\bm{x}_u^{1,\text{out}} \in \mathbb{R}^{L_1 \times d_1}$. Patch merging then maps $\bm{x}_u^{1,\text{out}}$ to $\bm{x}_u^2 \in \mathbb{R}^{L_2 \times d_2}$, where $L_2 = L_1/4$ and $d_2 = 2d_1$.
At Stage 2, $\bm{x}_u^2$ is processed through the CUA module and $N_2$ Swin Transformer blocks to yield $\bm{x}_u^{2,\text{out}} \in \mathbb{R}^{L_2 \times d_2}$, which is then merged into $\bm{x}_u^3 \in \mathbb{R}^{L_3 \times d_3}$, where $L_3 = L_2/4$ and $d_3 = 2d_2$.
Similarly, Stage 3 processes $\bm{x}_u^3$ through the CUA module and $N_3$ Swin Transformer blocks to obtain $\bm{x}_u^{3,\text{out}} \in \mathbb{R}^{L_3 \times d_3}$, and patch merging further produces $\bm{x}_u^4 \in \mathbb{R}^{L_4 \times d_4}$, where $L_4 = L_3/4$ and $d_4 = 2d_3$.
Finally, at Stage 4, $\bm{x}_u^4$ is processed by the CUA module and $N_4$ Swin Transformer blocks to produce the final encoder output $\bm{x}_u^{4,\text{out}} \in \mathbb{R}^{L_4 \times d_4}$.
After the final stage, the token sequence $\bm{x}_u^{4,\text{out}}$ is aggregated into a compact semantic feature vector $\bm{z}_u \in \mathbb{R}^{d_4}$, which is then normalized to form the transmitted signal $\bm{s}_u = \sqrt{\frac{d_4}{\mathbb{E}[\|\bm{z}_u\|^2_2]}} \bm{z}_u, u \in \{i,j\}$.
The BS superimposes the signals of paired users $i$ and $j$ as:
\begin{align}
    \bm{s}_k^{(i,j)}= \sqrt{\frac{p_k}{2}} (\bm{s}_i +  \bm{s}_j), \quad k \in \mathcal{K},
\end{align}
where $p_k$ is the transmit power for group $k$. The received signal at user $u$ is:
\begin{align}
\label{eq: received signal}
    \bm{y}^{(i,j)}_{k,u} =  {h}_{u} \bm{s}_k^{(i,j)} + \bm{n}_u, \quad u \in \{i,j\},
\end{align}
where $\bm{n}_u \sim \mathcal{N}(0, \sigma^2_u \bm{I})$ is the additive white Gaussian noise (AWGN) with the power $\sigma^2_u$, and ${h}_{u}$ is the channel gain from the BS to user $u$.
The signal ${\bm{y}}^{(i,j)}_{k,u}$ is then processed by the channel decoder to output $\hat{\bm{y}}^{(i,j)}_{k,u}$ and the semantic decoder $\bm{D}_{{\bm{\phi}_u}}(\cdot)$ with trainable parameters $\bm{\phi}_u$ to output $\hat{\bm{x}}^{(i,j)}_u$:
\begin{align}
\label{eq: decoder}
    \hat{\bm{x}}^{(i,j)}_u = \bm{D}_{\bm{\phi}_u}(\hat{{\bm{y}}}^{(i,j)}_{k,u}), \quad u \in \{i,j\}.
\end{align}

\subsection{The Mechanism of the Proposed CUA Module}
\vspace{-3mm}
The CUA module is inserted at each stage $\ell$ of the SwinJSCC encoder to enable semantic feature interaction between paired users. 
Let $\bm{x}_i^\ell, \bm{x}_j^\ell \in \mathbb{R}^{L_\ell \times d_\ell}$ denote the token sequences of users $i$ and $j$ at stage $\ell$. 
They are partitioned into $N_W^\ell = L_\ell / M^2$ non-overlapping patches, and each patch can be considered as a window with size $M \times M$. 
The CUA module updates $\bm{x}_i^\ell$ through a two-step cross-attention process:
\vspace{-6mm}
\begin{subequations}
\begin{align}
\check{\bm{x}}_i^\ell &= \mathcal{W}_{i \leftarrow j}(\mathcal{L}(\bm{x}_i^\ell), \bm{x}_j^\ell) + \bm{x}_i^\ell, \\
\bar{\bm{x}}_i^\ell &= \mathcal{M}(\mathcal{L}(\check{\bm{x}}_i^\ell)) + \check{\bm{x}}_i^\ell, \\
\tilde{\bm{x}}_i^\ell &= \mathcal{T}_{i \leftarrow j}(\mathcal{L}(\bar{\bm{x}}_i^\ell), \bar{\bm{x}}_j^\ell) + \bar{\bm{x}}_i^\ell, \\
\bm{x}_i^{\ell,\text{out}} &= \mathcal{M}(\mathcal{L}(\tilde{\bm{x}}_i^\ell)) + \tilde{\bm{x}}_i^\ell,
\end{align}
\end{subequations}
where $\mathcal{L}(\cdot)$ denotes LayerNorm, $\mathcal{M}(\cdot)$ is a multilayer perceptron, 
$\mathcal{W}(\cdot)$ applies windowed multi-head cross-attention inside non-overlapping $M \times M$ windows between $\bm{x}^\ell_i$ and $\bm{x}^\ell_j$ to capture the useful semantic information from user $j$ to user $i$, and 
$\mathcal{T}_{i \leftarrow j}(\cdot,\cdot)$ performs the shifted-window cross-attention to capture cross-window dependencies from user $j$ to user $i$. 
Specifically, for each window $m \in \{1, \dots, N_W^\ell\}$, the cross-attention outputs are computed as:
\begin{subequations}
\begin{align}
\mathcal{W}_{i \leftarrow j}^m &= \gamma \cdot g_m \cdot \mathcal{A}(\bm{x}_{i,m}^\ell, \bm{x}_{j,m}^\ell), \\
\mathcal{T}_{i \leftarrow j}^m &= \gamma \cdot g_m \cdot \mathcal{S}^{-1}\!\left[\mathcal{A}(\mathcal{S}(\bm{x}_{i,m}^\ell), \mathcal{S}(\bm{x}_{j,m}^\ell))\right],
\end{align}
\end{subequations}
where $\mathcal{A}(\cdot)$ denotes multi-head cross-attention with queries from $\bm{x}_{i,m}^\ell$ and keys/values from $\bm{x}_{j,m}^\ell$, and $\mathcal{S}(\cdot)$ is the window-shifting operation \cite{liu2021swin}. 
To control the degree of semantic information transfer from user $j$ to user $i$, a cosine-similarity gate $g_m \in [0,1]$ is applied to each window $m$ with a learnable residual scale $\gamma$, defined as 
$g_m = \sigma(\kappa \cdot \cos(\bm{x}^\ell_{i,m}, \bm{x}^\ell_{j,m}))$, 
where $\kappa>0$ is a hyperparameter, $\sigma(\cdot)$ is the sigmoid function, and 
$\cos(\bm{x}^\ell_{i,m}, \bm{x}^\ell_{j,m}) = \frac{1}{M^{2}}\sum_{t=1}^{M^{2}}\frac{\langle \bm{x}^{\ell}_{i,m}[t],\, \bm{x}^{\ell}_{j,m}[t] \rangle}{\|\bm{x}^{\ell}_{i,m}[t]\|_{2}\, \|\bm{x}^{\ell}_{j,m}[t]\|_{2}}$ 
is the averaged token-wise cosine similarity between $\bm{x}^\ell_{i,m}$ and $\bm{x}^\ell_{j,m}$. 
After the CUA computation, patch merging maps $(\bm{x}_i^{\ell,\text{out}}, \bm{x}_j^{\ell,\text{out}})$ to $(\bm{x}_i^{\ell+1}, \bm{x}_j^{\ell+1})$ with $L_{\ell+1}=L_\ell/4$ and $d_{\ell+1}=2d_\ell$.

\vspace{-2.5mm}
\subsection{Data Transmission Model} 
\vspace{-1.5mm}

We introduce a binary decision variable to elaborate the user pairing structure. 
For each user $i$, we define a user pairing vector $\boldsymbol{w}_{i} = [\boldsymbol{w}_{i1}, \cdots, \boldsymbol{w}_{ij}, \cdots, \boldsymbol{w}_{iN}]$ with $\boldsymbol{w}_{ij} = [w^1_{ij}, \cdots, w^k_{ij}, \cdots, w^K_{ij}]$ being a user group allocation vector for user pair $(i,j)$, and with $w^k_{ij} \in \{0,1\}$ being group allocation index. In particular, $w^k_{ij}=1$ implies that users $i$ and $j$ are paired and allocated to group $k$, and $w^k_{ij}=0$, otherwise.
Correspondingly, for each user group $k$, we define the bandwidth vector $\boldsymbol{b} = [b_{1}, \cdots, b_{K}]$, where $b_{k}$ is the shared bandwidth in group $k$.
We also define the power allocation vector $\boldsymbol{p} = [p_1, \cdots, p_K]$.
Thus, the data rate of the BS transmitting data to the user $i$ is:
\begin{align}
\label{eq: data rate S}
    r_{i}(\boldsymbol{w}_{i}, \boldsymbol{b}) = \sum_{j\in\mathcal{N} \setminus  \{i\}} \sum^K_{k=1} w^k_{ij}  b_{k} \log_2 \left(1 + \frac{|h_i|^2 p_k}{ 2N_0b_{k} + |h_i|^2 p_k}\right), 
\end{align}
where $N_0$ is the AWGN power spectral density.

\vspace{-2mm}
\subsection{Time Consumption Model}
\vspace{-1.5mm}
The time that each user group $k\in \mathcal{K}$ consumes for end-to-end communication consists of three parts: 1) the delay of model calculation at the BS; 2) the delay of data transmission, and 3) the delay of model calculation at each user's receiver, which are specified as follows. 

\subsubsection{Delay of model calculation at the BS}
The time consumption $\tau^{\text{BS}}_{i}$ of user $i$ for model calculation is denoted as:
\begin{align}
\label{eq: BS calculation time}
    \tau^{\text{BS}}_{i} =  \frac{\chi_{\text{BS}} q_i \Gamma(\boldsymbol{\theta}_i)}{f_{\text{BS}}}, \forall i \in \mathcal{N},
\end{align}
where $f_{\text{BS}}$ is the frequency of the central processing unit (CPU) clock of the BS, $\chi_{\text{BS}}$ is the number of CPU cycles required for computing data (in bits) at the BS side, $q_i$ is the data size of the image $\boldsymbol{x}_i$ that the BS needs to extract semantic information, and $\Gamma(\boldsymbol{\theta}_i)$ is the size of the user $i$'s encoder with the total encoder parameter $\boldsymbol{\theta}_i$. Thus, the total model calculation time at the BS for user group $k$ is
\begin{align}
    \tau^{\text{BS}}_k (\boldsymbol{w}) 
= \sum_{1\le i<j\le N} w^k_{ij}\, (\tau^{\text{BS}}_{i} + \tau^{\text{BS}}_{j}).
\end{align}

\vspace{-2mm}
\subsubsection{Delay of data transmission} 
Since the dimensions of the encoder output are fixed for all users, the data size that needed to transmit to each user is constant, and can be denoted by $Q$. The time that BS transmits semantic information to user $i$ is
\begin{align}
\label{eq: transmission time}
    t_{i}(\boldsymbol{w}_{i}, \boldsymbol{b}) = \frac{Q}{r_{i}(\boldsymbol{w}_{i}, \boldsymbol{b})},
\end{align}
where $Q$ is the data size of the output of each encoder. Thus, the total transmit time for user group $k$ is
\begin{align}
    t_k(\boldsymbol{w}, \boldsymbol{b}) 
= \sum_{1\le i<j\le N} w^k_{ij}\, \max \{t_{i}(\boldsymbol{w}_{i}, \boldsymbol{b}),\, t_{j}(\boldsymbol{w}_{j}, \boldsymbol{b})\}.
\end{align}

\subsubsection{Delay of model calculation at each receiver}
The time consumption $\tau^{\text{Rx}}_{i}$ at user $i$'s receiver for model calculation is:
\begin{align}
\label{eq: receiver calculation time}
    \tau^{\text{Rx}}_{i} =  \frac{\chi_{i} Q \Gamma(\boldsymbol{\phi}_i)}{f_{i}}, \forall i \in \mathcal{N}.
\end{align}
where $f_{i}$ is the frequency of the CPU clock of the user $i$, and $\chi_{i}$ is the number of CPU cycles required for computing data (in bits) at the user $i$'s receiver side, and $\Gamma(\boldsymbol{\phi}_i)$ is the size of the user $i$'s decoder with the parameter $\boldsymbol{\phi}_i$.
Thus, the total time consumption of model calculation at receivers side of the user group $k \in \mathcal{K}$ is calculated as:
\begin{align}
    \tau^{\text{Rx}}_k(\boldsymbol{w})
= \sum_{1\le i<j\le N} w^k_{ij}\, (\tau^{\text{Rx}}_{i}+\tau^{\text{Rx}}_{j}).
\end{align}
Therefore, the total time consumption of user group $k$ is:
\begin{align}
    T_k(\boldsymbol{w}, \boldsymbol{b}) = \tau^{\text{BS}}_k(\boldsymbol{w}) + t_k(\boldsymbol{w}, \boldsymbol{b}) +\tau^{\text{Rx}}_k(\boldsymbol{w}).
\end{align}

\vspace{-5mm}
\subsection{Energy Consumption Model}
The energy consumption of end-to-end system from the BS to user group $k \in \mathcal{K}$ consists of two components: data computing energy consumption, and data transmission energy consumption,
which is illustrated as follows.
\vspace{-1mm}
\subsubsection{Data computing energy consumption}
The energy consumption of the BS extraction semantic information for user $i$ is expressed as:
$\epsilon^{\text{BS}}_i = \zeta_{\text{BS}} f_{\text{BS}}^2 \chi_{\text{BS}} q_i \Gamma(\boldsymbol{\theta}_i)$,
where $\zeta_{\text{BS}}$ is the BS energy consumption coefficient.
The energy consumption of user $i$ recovering original images can be expressed as
$\epsilon^{\text{Rx}}_i = \zeta_i f_i^2 \chi_{i} Q \Gamma(\boldsymbol{\phi}_i)$,
where $\zeta_i$ is the user device energy consumption coefficient.
Therefore, the total data computing energy consumption of user group $k$ is calculated by
\begin{align}
   \epsilon_k(\boldsymbol{w}) 
= \sum_{1\le i<j\le N} w^k_{ij}\,(\epsilon^{\text{BS}}_i+\epsilon^{\text{BS}}_j + \epsilon^{\text{Rx}}_i+\epsilon^{\text{Rx}}_j).
\end{align}

\subsubsection{Data transmission energy consumption}
The energy consumption of data transmission from the BS to user group $k$ is calculated by $e_k(\boldsymbol{w}, \boldsymbol{b}) = p_k t_k(\boldsymbol{w}, \boldsymbol{b})$.
Thus, the total energy consumption $E_k$ of user group $k\in \mathcal{K}$ is calculated by:
\begin{align}
    E_k(\boldsymbol{w}, \boldsymbol{b}) = \epsilon_k(\boldsymbol{w}) + e_k(\boldsymbol{w}, \boldsymbol{b}).
\end{align}

\subsection{The Distortion Model}
The distortion function of user $i$ within user pair $(i,j)$ is:
\begin{align}
\label{eq: distortion function}
    D_i(\boldsymbol{w}_{ij}) = \sum^{K}_{k=1} w^k_{ij} \|\hat{\boldsymbol{x}}^{(i,j)}_i-\boldsymbol{x}_i\|^2_2, \forall i \ne j, i,j \in \mathcal{N},
\end{align}
where $\|\hat{\boldsymbol{x}}^{(i,j)}_i-\boldsymbol{x}_i\|^2_2$ is the MSE between the reconstructed image $\hat{\boldsymbol{x}}^{(i,j)}_i$ and the original image $\boldsymbol{x}_i$ when user $i$ is paired with $j$ to implement SFMA. 
We consider a reliability-guaranteed physical layer (e.g., sufficiently strong channel coding and HARQ) such that the residual decoding error probability is negligible at the target operating point. 

\subsection{Problem Formulation}
Our objective is to minimise the aggregate semantic distortion across all users with respect to system budgets on bandwidth, end‑to‑end latency, and energy consumption.
The optimisation problem is formulated as:
\begin{subequations}
\label{eq: overall problem}
\begin{align}
    & \underset{\boldsymbol{w}, \boldsymbol{b}}{\min}  \sum^{N}_{i=1} \sum_{j\in\mathcal{N} \setminus  \{i\}} D_i(\boldsymbol{w}_{ij}), \tag{17} \\
    &\text{s.t.} \; D_i(\boldsymbol{w}_{ij}) \leq D^{\max}, \forall i \ne j , i,j \in \mathcal{N},  \label{eq: quality requirement} \\
    &  \quad  0 \leq T_{k}(\boldsymbol{w}, \boldsymbol{b}) \leq T^{\max}, \forall k \in \mathcal{K}, \label{eq: time constraint} \\
    &  \quad  \sum^{K}_{k=1} E_k(\boldsymbol{w}, \boldsymbol{b}) \leq E^{\max},  \label{eq: energy constraint} \\
    &  \quad   \sum_{k=1}^{K} {b_{k}}  \leq B^{\max}, b_k \geq 0, \forall k \in \mathcal{K}, \label{eq: bandwidth constraint} \\
    &  \quad  \sum_{i=1}^N\sum^N_{j=i+1}w^k_{ij}=1, \forall k \in \mathcal{K}, \label{eq: w decision variable constraint 1} \\
    &  \quad \sum_{j \ne i, j=1}^N\sum^K_{k=1}w^k_{ij}=1, \forall i \in \mathcal{N}, \label{eq: w decision variable constraint 2} \\
    &  \quad w^k_{ij} \in\{0,1\}, \forall k \in \mathcal{K}, \forall i,j \in \mathcal{N}, \label{eq: nonnegative}
\end{align}
\end{subequations}
where $B^{\max}$ is the maximum bandwidth, $T^{\max}$ is the maximum time delay, $E^{\max}$ is the maximum energy budget. 

\vspace{-2.5mm}
\section{Algorithm Design}
\vspace{-2mm}
In this section, we propose an optimisation algorithm to address the problem \eqref{eq: overall problem}, which contains two sub-problems: a user-pairing sub-problem and a bandwidth allocation sub-problem.
The constraint \eqref{eq: time constraint} is simplified as:
\begin{align}
\label{eq: subproblem1 constraint time}
    t_k(\boldsymbol{w}, \boldsymbol{b}) \leq \sum^{N}_{i=1} \sum_{j\in\mathcal{N} \setminus  \{i\}} w^k_{ij} \Delta_{ij}, \forall k \in \mathcal{K}, 
\end{align}
where $\Delta_{ij} \triangleq T^{\max} - \tau^{\text{BS}}_i - \tau^{\text{Rx}}_i- \tau^{\text{BS}}_j - \tau^{\text{Rx}}_j, \forall i,j \in \mathcal{N}$.
Assume that the user pair $(i,j)$ is assigned to user group $k$ with user $i$ and $j$, i.e., $w^k_{ij} = 1$ and $w^n_{ij} = 0$ for all $n \ne k$. 
For user group $k$ with user $i$ and $j$, the constraint \eqref{eq: subproblem1 constraint time} can be equivalently transformed to the following inequality:
\begin{align} \label{eq: F_i inequality}
    F_u(b_k) \geq \frac{Q}{\Delta_{ij}}, \forall u \in \{i,j\},
\end{align}
where
\begin{align} \label{eq: F_i}
    F_u(b_k) \triangleq b_k \log_2 \left(1 + \frac{|h_u|^2p_k}{2N_0b_k+|h_u|^2p_k} \right).
\end{align}
Based on the above definition, we first establish the following \text{\textbf{Lemma}} to characterise the monotonicity property of $F_u(b_k)$, which will later be used to derive the minimum feasible bandwidth requirement.
\vspace{-3mm}
\begin{lemma} \label{lemma: monotony}
    For any given $p_k >0$ and $u \in \{i,j\}$, $F_u(b_k)$ is strictly increasing with respect to $b_k >0$. Moreover, $F_u(b_k)$ is strictly concave with $\lim_{b_k \to \infty}F_u(b_k)=\frac{|h_u|^2p_k}{2N_0\ln2}$. Hence, \eqref{eq: F_i inequality} admits a unique root $b_u^{\min}(\Delta_{ij})$ if and only if  $\frac{Q}{\Delta_{ij}} < \frac{|h_u|^2p_k}{2N_0\ln2}$.
\end{lemma}

\proof Since $F_u'(b_k)= \log_{2}\!\left(\frac{2N_0 b_k + 2|h_u|^{2} p_k}{2N_0 b_k + |h_u|^{2} p_k}\right)- \frac{2N_0 b_k|h_u|^{2} p_k}{\ln 2(2N_0 b_k + 2|h_u|^{2} p_k)(2N_0 b_k + |h_u|^{2} p_k)} > 0$, and $F''_u(b_k) = -\frac{2N_0 |h_u|^{2} p_k}{\ln 2}\cdot \frac{2N_0 b\,(2N_0 b+2|h_u|^{2} p_k)+(2N_0 b+|h_u|^{2} p_k)^2} {(2N_0 b+2|h_u|^{2} p_k)^2(2N_0 b+|h_u|^{2} p_k)^2} <0$, $F_u(b_k)$ is strictly concave and increasing with respect to $b_k$. Hence, $\lim_{b_k \to \infty} F_{u}(b_k) = \frac{|h_u|^2p_k}{2N_0\ln2}$.

According to \text{\textbf{Lemma}} \ref{lemma: monotony}, the inequality \eqref{eq: F_i inequality} exists exactly one root $b^{\min}_u$ if and only if $\frac{Q}{\Delta_{ij}} < \frac{|h_u|^2p_k}{2N_0\ln2}, \forall u \in \{i,j\}$. Consequently, the joint latency constraint for user pair $(i,j)$ can be translated into a minimum required bandwidth as
\begin{align}
    b^{\min}_{ij} = \max \{b^{\min}_i(\Delta_{ij}), b^{\min}_j(\Delta_{ij})\}.
\end{align}
Accordingly, the delay constraint \eqref{eq: subproblem1 constraint time} is reformulated as a bandwidth lower-bound condition \eqref{eq: b_k constraint 1}
\begin{align} \label{eq: b_k constraint 1}
    b_k \geq \sum_{1\leq i <j \leq N} w^k_{ij}b^{\min}_{ij}, \forall k \in \mathcal{K},
\end{align}
which provides a tractable convex form for subsequent feasibility verification.

After reformulating the delay constraint into an explicit bandwidth lower bound in \eqref{eq: b_k constraint 1}, we next address the energy constraint \eqref{eq: energy constraint}, which can be equivalently written as
\begin{align} \label{eq: b_k constraint 2}
    \sum^N_{i=1} \sum^{N}_{j=i+1} \sum^K_{k=1} w^k_{ij} p_k \xi_k(b_k) \leq E^{\max} - E^{\text{const}},
\end{align}
where $\xi_k(b_k) \triangleq \max \{\frac{Q}{F_i(b_k)},\frac{Q}{F_j(b_k)}\}$, and $E^{\text{const}} = \sum^N_{i=1}[ \zeta_{\text{BS}}f_{\text{BS}}^2q_i\Gamma(\boldsymbol{\theta}_i)+\zeta_i(f_i)^2 Q \Gamma(\boldsymbol{\phi}_i)] $ is a constant. Let $t_u(b_k) =\frac{Q}{F_u(b_k)}, \forall u \in \{i,j\}$. From \textit{\textbf{Lemma}} \ref{lemma: monotony}, we know that $t_u(b_k)$ is monotonically decreasing when $b_k >0, \forall k \in \mathcal{K}$.

Problem \eqref{eq: overall problem} reveals that the objective depends solely on the matching variables $\boldsymbol{w}_i$, while the continuous bandwidth vector $\boldsymbol{b}$ only affects the feasibility of a given matching. That is, for any candidate matching satisfying the distortion constraints in \eqref{eq: quality requirement}, the bandwidth allocation $\boldsymbol{b}$ is optimized merely to verify whether constraints \eqref{eq: energy constraint}, \eqref{eq: bandwidth constraint}, and \eqref{eq: subproblem1 constraint time} can be jointly met. If no feasible $\boldsymbol{b}$ exists, the matching is deemed infeasible. Therefore, the original mixed-integer problem can be equivalently decomposed into: (i) a minimum-weight perfect matching (MWPM) problem over $\boldsymbol{w}_i$ to determine the distortion-optimal pairing, and (ii) a convex bandwidth allocation subproblem that acts as a feasibility check under delay, energy, and bandwidth constraints.

\setlength{\columnsep}{0.02 in}

To include \eqref{eq: quality requirement} in the objective function, we define:
\begin{align} \label{eq: edge weight}
C_{ij} =
    \begin{cases}
    d_{ij}, & \text{if } D_i(\boldsymbol{w}_{ij})\leq D^{\max}, D_j(\boldsymbol{w}_{ij})\leq D^{\max}, \\
    +\infty, & \text{otherwise},
    \end{cases}
\end{align}
where $d_{ij} = \|\hat{\boldsymbol{x}}^{(i,j)}_i-\boldsymbol{x}_i\|^2_2 + \|\hat{\boldsymbol{x}}^{(i,j)}_j-\boldsymbol{x}_j\|^2_2$, $\forall i \ne j, i,j \in \mathcal{N}$. We define a complete graph with nodes $\mathcal{N}$ and edge weight \eqref{eq: edge weight}. The pairing variable $\boldsymbol{w}$ is a $0$-$1$ indicator of a perfect matching satisfying constraints \eqref{eq: w decision variable constraint 1}-\eqref{eq: nonnegative}. Then, problem \eqref{eq: overall problem} can be reformulated as a standard MWPM problem in \eqref{eq: subsubproblem 1} plus feasibility verification in \eqref{eq: bandwidth allocation sub-problem}, which is denoted as:
\begin{subequations}
\label{eq: subsubproblem 1}
\begin{align}
     \min_{\boldsymbol{w}}& \sum^N_{i=1}\sum^N_{j=i+1}\sum^K_{k=1}w^k_{ij}C_{ij}   \tag{25} \\
     \text{s.t.} \; & \eqref{eq: w decision variable constraint 1}, \eqref{eq: w decision variable constraint 2},
\end{align}
\end{subequations}
and a bandwidth allocation sub-problem:
\begin{subequations}
\label{eq: bandwidth allocation sub-problem}
\begin{align}
    \min_{\boldsymbol{b}}&  \sum^K_{k=1} f_k(b_k) \tag{26} \\
    \text{s.t.} \; & \eqref{eq: bandwidth constraint}, \eqref{eq: b_k constraint 1}, \eqref{eq: b_k constraint 2},
\end{align}
\end{subequations}
where $f_k(b_k) \triangleq \sum_{1\leq i < j \leq N}w^k_{ij}p_k\xi_k(b_k)$.
We solve the MWPM problem on a complete undirected graph with node set $\mathcal{N}$ and edge weights defined in \eqref{eq: edge weight}. The MWPM problem is solved by an Edmonds‑type Blossom algorithm \cite{kolmogorov2009blossom} in $\mathcal{O}(N^3)$ time.
Problem \eqref{eq: subsubproblem 1} only determines the pairing method $\boldsymbol{w}$ that results in the least distortion. The sub-problem \eqref{eq: bandwidth allocation sub-problem} is only used to verify whether all constraints can be simultaneously satisfied under the given $\boldsymbol{w}$. If the sub-problem \eqref{eq: bandwidth allocation sub-problem} is not feasible, then proceed to the next candidate $\boldsymbol{w}$.
Problem \eqref{eq: bandwidth allocation sub-problem} is convex, and it can be solved by \textit{\textbf{Lemma}} \ref{lemma: bandwidth allocation}. 
\vspace{-3mm}
\begin{lemma} \label{lemma: bandwidth allocation}
Given a matching $\boldsymbol{w}$, problem \eqref{eq: bandwidth allocation sub-problem} is convex. Its optimal solution satisfies
\begin{align}
    b_k^{\star} = 
    \max\!\left\{L_k(\boldsymbol{w}),\,\tilde{b}_k\!\left(\Theta^{\star}\right)\right\},
    \quad \forall k \in \mathcal{K},
\end{align}
where $L_k(\boldsymbol{w}) = \sum_{1 \le i < j \le N} w_{ij}^k b_{ij}^{\min}$ denotes the minimum required bandwidth for the user group $k$, and
\begin{align}
    \tilde{b}_k(\Theta) =
\begin{cases}
G^{-1}_{i,k}(\Theta), & \text{if } \dfrac{Q}{F_i(G^{-1}_{i,k}(\Theta))} \geq \dfrac{Q}{F_j(G^{-1}_{i,k}(\Theta))}, \\
G^{-1}_{j,k}(\Theta), & \text{otherwise}.
\end{cases}
\end{align}
Here, $G_{u,k} \triangleq \frac{p_kQF'_u(b_k)}{F^2_u(b_k)}$ is strictly decreasing with respect to $b_k>0$ and thus invertible, and $G^{-1}_{u,k}(\cdot)$ is the inverse of $G_{u,k}(\cdot)$. $\Theta^{\star} \in (0,\Theta^{\max}]$ is the Lagrange multiplier associated with the total bandwidth constraint \eqref{eq: bandwidth constraint} determined by the bisection method so that $\sum_{k=1}^{K}b^\star_k(\Theta^\star)= B^{\max}$, with $\Theta^{\max} = \max_{1 \le k \le K} \{G_{i,k}(L_k), G_{j,k}(L_k)\}$.
\end{lemma}
\proof Please refer to the Appendix \ref{appendix: proof of lemma 2}.

Therefore, the overall problem \eqref{eq: overall problem} can be solved by Algorithm \ref{algorithm: user pairing 1}.

\begin{algorithm}[!t]
\caption{User Pairing Algorithm under Bandwidth Constraints}
\begin{algorithmic}[1]
\State \textbf{Initialization:} a fixed power allocation vector $\boldsymbol{p}$.
\State Calculate $\{C_{ij}\}, \forall i\ne j, i,j \in\mathcal{N}$ according to \eqref{eq: edge weight}.
\State Use a Blossom MWPM problem solver \cite{kolmogorov2009blossom} to generate the top-$W$ candidate pairing schemes $\{\boldsymbol{w}^{(1)},\dots,\boldsymbol{w}^{(W)}\}$ in ascending order of the total distortion.
\Repeat
    \For{$t=1:W$}
        \State For $\boldsymbol{w}^{(t)}$, solve the bandwidth allocation sub-problem \eqref{eq: bandwidth allocation sub-problem} according to Lemma \ref{lemma: bandwidth allocation} to obtain $\boldsymbol{b}^{(t)}$.
        \If{\eqref{eq: bandwidth allocation sub-problem} is \textbf{feasible}}
            \State Set $(\boldsymbol{w}^\star,\boldsymbol{b}^\star)=(\boldsymbol{w}^{(t)},\boldsymbol{b}^{(t)})$; \textbf{break}.
        \EndIf
    \EndFor
\Until a feasible pairing–bandwidth combination $(\boldsymbol{w}^\star,\boldsymbol{b}^\star)$ is found.
\State \textbf{Output:} $(\boldsymbol{w}^\star,\boldsymbol{b}^\star)$.
\end{algorithmic}
\label{algorithm: user pairing 1}
\end{algorithm}

\vspace{-4mm}
\section{Simulation Results}
\vspace{-2mm}
\label{sec: results} 
In this section, we evaluate the proposed SFMA framework and optimization algorithm. We simulate a downlink system with $N=16$ users uniformly distributed in a $500\ \text{m} \times 500\ \text{m}$ area, served by a central BS. The path loss follows $128.1+37.6\log_{10}d$ ($d$ is in $\text{km}$), and the standard deviation of shadow fading is $8\ \text{dB}$. Experiments are performed on ImageNet-100 images of size $3\times256\times256$, with system budgets: $B^{\max}\in[5,40]\ \text{MHz}$, $p_1=\cdots= p_K = 1 \text{W}$, and $E^{\max}=200\ \text{J}$. The proposed method is compared against four baselines: 1) Users are randomly paired, and each group receives equal bandwidth (Random + EqualBW); 2) User pairs are formed by greedily selecting the smallest distortion entries in the distortion matrix $C_{ij}$, followed by equal bandwidth allocation (Greedy + Equal); 3) Users are sorted by channel gain, and each strong user is paired with a weak user. Bandwidth is evenly divided among all groups (Channel-Balanced + EqualBW); 4) Users are randomly paired, but bandwidth is optimally allocated using the closed-form solution derived from the KKT conditions in Lemma \eqref{lemma: bandwidth allocation} (Random + KKT).

\addtolength{\topmargin}{0.02in}

Fig. \ref{fig: PSNR} compares the PSNR performance of DeepJSCC, SFMA without CUA (w/o CA), and SFMA with CUA (w/ CA) under random and similar pairing over an AWGN channel. The results demonstrate that w/ CA consistently outperforms w/o CA across all SNRs, with the largest gain observed in the mid-SNR regime. This gain stems from the CUA module's ability to selectively share semantically relevant features via a cosine-similarity gate, which is most effective when the channel is neither overly noisy nor saturated. Furthermore, similar pairing yields higher PSNR than random pairing, as stronger inter-image correlations enable more informative feature exchange. The nearly overlapping curves of U1 and U2 reflect the near-symmetric effective SNRs achieved by SFMA’s equal-power superimposition within each pair.

\begin{figure}[t]
    \centerline{\includegraphics[width=0.4\textwidth]{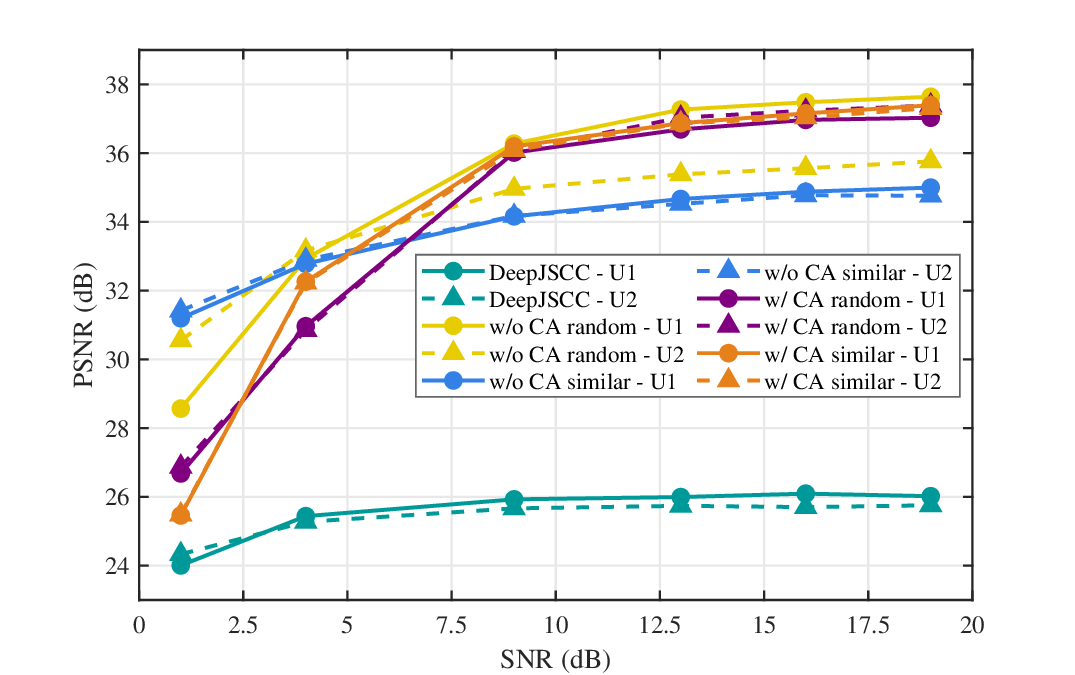}}
    \caption{PSNR versus SNR over AWGN channel.}
    \label{fig: PSNR}
\end{figure}
\vspace{-2mm}
\begin{figure}[t]
    \centerline{\includegraphics[width=0.4\textwidth]{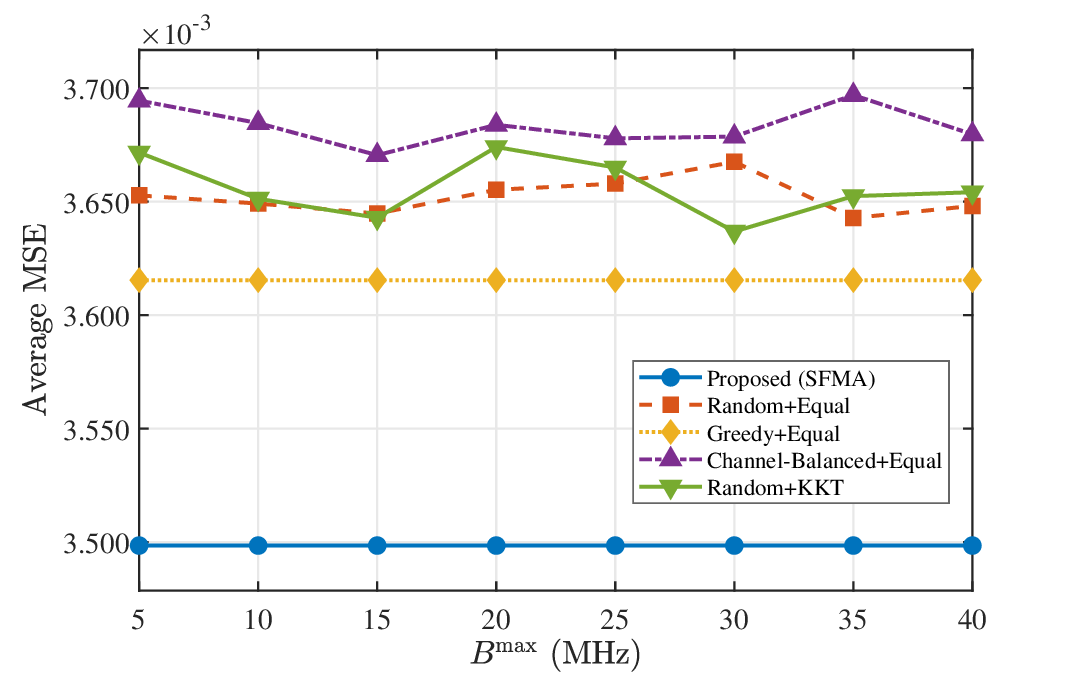}}
    \caption{Average MSE over various $B^{\max}$.}
    \label{fig: average MSE}
\end{figure}

Fig. \ref{fig: average MSE} depicts the average MSE versus the maximum bandwidth $B^{\max}$. The proposed SFMA optimisation achieves the lowest MSE across all bandwidths, outperforming all baselines. Compared to Random+EqualBW and Greedy+EqualBW, it reduces distortion by about $5 \text{--} 10\%$ under moderate bandwidth budgets, owing to its joint use of a precomputed distortion table and convex bandwidth-feasibility region to allocate resources only to feasible and distortion-sensitive pairs. While Channel-Balanced+EqualBW and Random+KKT show moderate gains, they are still inferior due to the lack of joint pairing and resource optimisation. The superiority of our approach arises from the MWPM-based pairing, which minimizes semantic distortion while satisfying latency and energy constraints.

\section{Conclusion} \label{sec: conclusion}
This paper proposed a SFMA framework for semantic-aware downlink multi-user communications, integrating a CUA module to adaptively fuse and superimpose correlated semantic features while preserving reconstruction quality. We further developed a two-stage optimization strategy for user pairing and resource allocation, formulated as a constrained MWPM problem, and established convex bandwidth regions along with a closed-form allocation policy based on inverse gradients. Simulations confirm significant distortion reduction, highlighting SFMA’s potential for 6G semantic communication. Future work will extend SFMA to multi-hop, mobile, and multi-group scenarios using reinforcement learning.

\appendices
\section{Proof of \textbf{Lemma} \ref{lemma: bandwidth allocation}} \label{appendix: proof of lemma 2}
The Lagrange function of problem \eqref{eq: bandwidth allocation sub-problem} is denoted as:
\begin{align}
\label{eq: KKT}
    \mathcal{L}(\boldsymbol{b}, \Theta, \boldsymbol{\mu}) &= \sum_{k=1}^K f_k(b_k) + \Theta\left(\sum_{k=1}^Kb_k - B^{\max}\right) \notag \\
    &+ \sum_{k=1}^K \mu_k (L_k(\boldsymbol{w}) - b_k) ,
\end{align}
where $\boldsymbol{\mu} = [\mu_1, \cdots, \mu_K]$, $\Theta$ are non-negative Lagrange multipliers. 
Since the objective of problem \eqref{eq: bandwidth allocation sub-problem} is the same as the left end of \eqref{eq: b_k constraint 2} and the energy consumption constraint \eqref{eq: b_k constraint 2} is only used for feasibility check, its multiplier is not introduced in \eqref{eq: KKT} and Karush–Kuhn–Tucker conditions (KKT) derivation.
We define the active user index $u^\star(k) \in \{i,j\}$ such that $\xi_k(b_k) =\frac{Q}{F_{u^\star(k)}(b_k)}$. 
$\xi_k(b_k)$ is a convex function since it is the upper bound of two convex decreasing functions, and its set of sub-gradients at the intersection is the convex hull of the derivatives on both sides. Thus, each $u^\star(k) \in \{i,j\}$ satisfies KKT conditions.
From the stationary point condition, we can yield
\begin{align}
    - \left. \frac{\partial f_k}{\partial b_k} \right|_{b_k = b_k^{\star}}= G_{u^{\star}(k),k}(b_k^{\star}) = \Theta^{\star} - \mu_k^{\star}.
\end{align}
From the KKT condition $\mu^\star_k(L_k-b^\star_k) = 0$, we have two cases:
\begin{itemize}
    \item If $b^\star_k > L_k$, then $\mu^\star_k = 0$ and $G_{u^\star(k),k}(b^\star_k) = \Theta^\star$. Therefore, $b^\star_k = G^{-1}_{u^\star(k),k}(\Theta^\star)$.
    \item If $b^\star_{k} = L_k$, we have $G_{u^\star(k),k}(L_k) \geq \Theta^{\star}$.
\end{itemize}
Combining both cases gives $b^\star_k=\max\{L_k(\boldsymbol{w}\},\tilde{b}_k(\Theta^\star))$. The mapping $\sum_{k=1}^Kb^\star_k(\Theta)$ is decreasing with respect to $\Theta$, hence there exists a unique $\Theta^{\star} \in (0,\Theta^{\max}]$ such that $\sum_{k=1}^K b^\star_k = B^{\max}$, which can be found by the bisection method.

\bibliographystyle{ieeetr}
\bibliography{ref}
\end{document}